\documentclass[11pt]{article}
\usepackage{amsmath,amsfonts,amsthm,amssymb}
\newtheorem{TT}{Theorem}[section]
\newtheorem{DD}{Definition}[section]

\newtheorem{CC}{Corollary}[section]
\newtheorem{LL}{Lemma}[section]

\def\ket#1{\left|#1\right\rangle}

\begin{document}

\begin{titlepage}

\bigskip
\bigskip

\begin{center}
{\Large\bf Whittaker pairs for the Virasoro algebra\\
and the Gaiotto - BMT states}
\end{center}

\bigskip

\begin{center}
    {\large\bf
    Ewa Feli\'nska}\footnote{e-mail: letjenen@gmail.com
}, {\large\bf
    Zbigniew Jask\'{o}lski}\footnote{e-mail: jask@ift.uni.wroc.pl
}
and
{\large\bf
    Micha{\l} Koszto{\l}owicz}\footnote{e-mail: mikosz@ift.uni.wroc.pl
}
\\

\vskip 1mm
   Institute of Theoretical Physics\\
 University of Wroc{\l}aw\\
 pl. M. Borna, 50-204 Wroc{\l}aw,
 Poland
\end{center}
\vskip .5cm
\begin{abstract}
In this paper we analyze  Whittaker modules
for two families of Wittaker pairs
related to the subalgebras of the Virasoro algebra generated by $L_r,\dots, L_{2r}$
and $L_1,L_n$. The structure theorems for the corresponding universal Whittaker modules
are proved and some of their consequences are derived.
All the Gaiotto \cite{Gaiotto:2009ma} and the Bonelli-Maruyoshi-Tanzini \cite{Bonelli:2011aa} states
in an arbitrary Virasoro algebra Verma module are explicitly constructed.
\end{abstract}
\end{titlepage}
\section{Introduction}

Whittaker
modules were first introduced by Arnal and Pinzcon \cite{Arnal74} in their study of the ${\rm sl}_2(\mathbb{C})$ algebra representations.
For an arbitrary complex infinite-dimensional semisimple
Lie algebra the theory of Whittaker modules was developed by Kostant \cite{Kostant78}.
The construction was based on the triangular decomposition $\mathfrak{g}=\mathfrak{n}^-\oplus \mathfrak{h}\oplus\mathfrak{n}^+$
and a  regular Lie algebra homomorphisms $\phi:\mathfrak{n}^+ \to \mathbb{C}$. For non-regular homomorphisms
the Kostant construction was analyzed  in \cite{McDowell85,McDowell93,Milicic97}. More recently
the Whittaker modules have been intensively investigated
for various
infinite-dimensional algebras with a triangular decomposition:
the Heisenberg and the affine Lie algebras \cite{Christodoupoulou08},
the generalized Weyl algebras \cite{Benkart09},
the Virasoro algebra \cite{Ondrus09,Ondrus11},
the twisted Heisenberg-Virasoro algebra \cite{Liu10},
the Schr\"{o}dinger-Witt algebra \cite{Zhang10},
the graded Lie algebras \cite{Wang09a},
the $W$-algebra $W(2,2)$ \cite{Wang09b}
and
the Lie algebras of Block type \cite{Wang09}.

A general categorial framework for Whittaker modules was proposed in \cite{Batra} where
the Kostant construction  was generalized
to a pair of Lie algebras $\frak{n}\subset\frak{g}$.
In this setting  the
Whittaker module is defined as a $\frak{g}$-module, generated
by a one-dimensional $\frak{n}$-invariant subspace.

{}

In two-dimensional conformal field theory the interest in Whittaker modules was stimulated by Gaiotto's paper
\cite{Gaiotto:2009ma}
on a particular versions of the AGT relation
\cite{Alday:2009aq} where on the CFT theory side
so called irregular blocks appear\footnote{
In the context of instanton counting
the Whittaker vectors of affine algebras were analyzed earlier in \cite{Braverman:2004vv,Braverman:2004cr}.}.
They are
defined as a scalar products of vectors of a Virasoro algebra
Verma module determined by the conditions
\begin{equation}
\label{wcond1}
L_1\ket{w} = \mu_1\ket{w}, \;\;L_2\ket{w} = \mu_2\ket{w}, \;\;L_n\ket{w} = 0, \;\;n>2.
\end{equation}
It has been conjectured in \cite{Gaiotto:2009ma} that not only these vectors but also
 higher order vectors
\begin{equation}
\label{wcond2}
L_r\ket{w} = \mu_r\ket{w},\;\dots,\;L_{2r}\ket{w} = \mu_{2r}\ket{w}, \;\;L_n\ket{w} = 0, \;\;n>2r,\;\;r>1
\end{equation}
 exist and are uniquely defined by the conditions above. The existence of the first order Gaiotto states
has been soon verified by an explicite
construction  in \cite{Marshakov:2009gn}.
Another construction in terms of the Jack symmetric polynomials was given in \cite{arXiv:1003.1049}.
The corresponding irregular blocks  were the simplest objects for which the AGT relation could be verified
\cite{Poghossian:2009mk,Hadasz:2010xp}.
In particular
the norm of the  first order Gaiotto state with $\mu_2=0$ corresponds to
the pure gauge partition function in four dimensions.
This relation has been  analyzed for various extensions of the AGT relations
with the CFT side described by
the $A_n$-Toda theories \cite{Wyllard:2009hg,Mironov:2009by,Taki:2009zd},
the $N=1$ super-symmetric Liouville theories \cite{Belavin:2011pp,Bonelli:2011kv,Ito:2011mw},
the para-$A_n$-Toda theories
\cite{Nishioka:2011jk,Wyllard:2011mn}
and
the general Toda theories \cite{Keller:2011ek}.

From the point of view of the general approach of \cite{Batra} the Gaiotto conditions (\ref{wcond2})
define a Whittaker vector of the Whittaker pair ${\cal V}_r\subset {\cal V}$ where
$\cal V$ is the Virasoro algebra and ${\cal V}_r$ its subalgebra generated by $L_r,\dots, L_{2r}$.
Recently a new type of coherent states corresponding to the Whittaker pairs
${\cal V}_{1,n}\subset {\cal V}$ where
 ${\cal V}_{1,n}$ is the subalgebra generated by $L_1,L_{n}$
has been introduced by Bonelli, Maruyoshi and Tanzini in
\cite{Bonelli:2011aa}. It was proposed that  the corresponding
irregular conformal blocks describe  partition functions of wild quiver gauge theories.

In the present paper we analyze  general algebraic properties of the Gaiotto and the BMT states for the Virasoro algebra.
This includes
the problem of the structure of Whittaker modules for the pairs ${\cal V}_r\subset {\cal V}$ and ${\cal V}_{1,n}\subset {\cal V}$
and the construction of corresponding states.
Since ${\cal V}_1={\cal V}_{1,2}$ the lower orders coincide.
In this case  the structure of Whittaker modules is already known \cite{Ondrus09,Ondrus11}
and the Gaiotto  states are constructed \cite{Marshakov:2009gn,arXiv:1003.1049}.
Our aim is to extend these results to higher orders. 

In Section 2 we use the method of \cite{Ondrus09}
to analyze  modules of the higher order Whittaker pairs ${\cal V}_r\subset {\cal V}$.
In view of CFT applications we restrict ourselves
to modules with a fixed central charge. Our main result is the structure theorem
for the universal Whittaker module of a general type. We say that a Lie algebra homomorphism  $\psi_r:{\cal V}_r\to {\cal V}$
is of the high rank if $\psi_r(L_{2r})\neq 0$ or $\psi_r(L_{2r-1})\neq 0$.
 For the high rank homomorphisms  the corresponding universal Whittaker module is simple.
In all other cases
 it has an infinite composition series with a single
composition factor uniquely determined by $\psi_r$.

Section 3 is devoted to the  Whittaker pairs ${\cal V}_{1,n}\subset {\cal V}$.
Most of the techniques developed in Section 2 can be applied in this case as well.
Note that for a non-regular Lie algebra homomorphisms $\psi_{1,n}:{\cal V}_{1,n}\to \mathbb{C}$
($\psi_{1,n}(L_1)=0$ or $\psi_{1,n}(L_n)=0$) we are back to the case ${\cal V}_r\subset {\cal V}$.
We give a complete analysis of Whittaker vectors in the universal Whittaker module in the cases $n=3,4$
and present some examples for higher orders.

In Section 4 we investigate
subspaces of Whittaker vectors of a given pair and type
in the Virasoro module of all anti-linear functionals on a Verma module.
Whenever the Shapovalov form is non-degenerate
the Gaiotto and the BMT states can  be recovered by "raising indices" level by level by the inverse to the Gram matrix.
For the first order states this is the
construction of \cite{Marshakov:2009gn}.
One can  define irregular blocks directly in terms of forms and the scalar product
determined level by level by the inverse Gram matrices. This approach conveniently separates the universal algebraic
properties
 from those specific to the weight dependent Shapovalov form.

It turns out that the first order Gaiotto states are uniquely (up to a scale factor) defined by the
conditions (\ref{wcond1}). This is no longer true for higher orders.
We present general representation  of the Gaiotto and the BMT states
in terms of systems of basic states.

The two types of  Whittaker pairs considered in this paper are special cases of a more
general pair ${\cal V}_{m,n}\subset {\cal V}$ where the subalgebra ${\cal V}_{m,n}$
is generated by
$$
L_m, L_n, L_{n+1}, \dots, L_{n+m-1},\;\;\; m<n.
$$
One can easily invent even more general subalgebras. If they contain the subalgebra ${\cal V}_r$
for a certain integer $r$ the method of \cite{Ondrus09} we have used in the present paper
can be applied to analyze the corresponding Whittaker modules. One may expect
new classes of simple modules of the Virasoro algebra. It is an interesting question to what extend
this general construction might be helpful in solving the problem of classification of all simple modules of the Virasoro algebra.
\footnote{A complete answer was recently found by Mazorchuk and Zhao in \cite{MazorchukZhao}.}

There are other possible continuations of the present work.
For the pair ${\cal V}_r\subset {\cal V}$ the structure theorem implies that all high rank Whittaker
modules of a given type are isomorphic. For the low rank modules the classification problem
remains open. Theorems of Section 2
provide a good starting point but general analysis is still to be done.

The abundance of higher order Gaiotto states rises the question of their classification.
One may expect \cite{Gaiotto:2009ma} that at least some of them arise as decoupling limits
of $n$-point conformal blocks. It would be interesting to find their algebraic characterization.
Similar questions arise for  the BMT states.
Finally, in view  of the developments mentioned above the extensions to other algebras would be desirable.

When the present paper was completed we became aware of a series of papers \cite{Guo1104,Guo1103,Guo1105}.
In \cite{Guo1104}
irreducible
 Virasoro algebra modules were studied. In particular
all  Whittaker modules for higher order pairs ${\cal V}_r\subset {\cal V}$ were explicitly constructed
by twisting oscillator representation of the twisted Heisenberg-Virasoro algebra.
The construction implies  that all high rank Whittaker modules
of an arbitrary order are simple (Theorem 7 of \cite{Guo1104}). In the present paper we obtain
the same result using different methods (Corollary 2.2).
The other two papers are devoted to the Whittaker modules of
the generalized Virasoro algebras \cite{Guo1103}
and of the Virasoro-like algebras \cite{Guo1105}.

We owe special thanks to Volodymyr Mazorchuk and Kaiming Zhao for
sending us their work \cite{MazorchukZhao} and in particular
for pointing out that Theorems 3.5, 3.6 and Corollary 3.7 of Section 3
in the first published version of our paper are valid only in the case of $n=3$.
The present version contains necessary corrections of Section 3.

.

\section{Whittaker pairs ${\cal V}_{r}\subset {\cal V}$ }

Let ${\cal V}$ be the Virasoro algebra  i.e.
${\cal V}={\rm span}_{\mathbb{C}} \left\{z,
L_n : n\in \mathbb{Z}
\right\}$
with the Lie bracket
\begin{eqnarray}
\label{virasoro}
\nonumber
\left[L_m,L_n\right] & = & (m-n)L_{m+n} +\frac{z}{12}m\left(m^2-1\right)\delta_{m+n},
\\
\left[L_m,z\right] & = & 0.
\end{eqnarray}
Let $r$ be a a positive integer. We introduce the subalgebra
$\mathcal{V}_{r}$
generated by $L_{r}, \hdots, L_{2r}$
 $$
 \mathcal{V}_{r}
  = \text{span}_{\mathbb{C}}\{L_r, L_{r+1}, \hdots \}\subset {\cal V}.
 $$
A
Lie algebra homomorphism
$\psi_r : \mathcal{V}_r \rightarrow \mathbb{C} $
is uniquely defined by its values on the algebra generators. This justifies the notation
$$
 \psi_r = \{\psi_r(L_r),\dots,\psi_r(L_{2r})\}.
$$
We consider only non trivial $\psi_r$ i.e. at least one of $\psi_r(L_i)$ does not vanish. For
non trivial homomorphisms   we define the rank:
$$
{\rm rank} \,\psi_r ={\rm max} \{r\leqslant i \leqslant 2r : \psi_r(L_i)\neq 0\}.
$$
From commutation relations (\ref{virasoro}) one has $\psi_r(L_k) = 0$ for all $k>{\rm rank} \,\psi_r$.
\begin{DD}
Let $V$ be a left $\mathcal{V}$-module,  $\psi_r : \mathcal{V}_r \rightarrow \mathbb{C} $
- a non trivial  Lie algebra homomorphism and $c$ - a complex number. A vector $\ket{w} \in V$ is called a Whittaker vector of
the Whittaker pair ${\cal V}_{r}\subset {\cal V}$,
the central charge $c$ and the type $\psi_r$
if
$$
z\ket{w} =c\ket{w}\;\;{\it and }\;\;L_k \ket{w} = \psi_r(L_k) \ket{w}\;\; {\it for} \;\;k\geqslant r.
$$
A $\mathcal{V}$-module $V$ is called a Whittaker module of the Whittaker pair ${\cal V}_{r}\subset {\cal V}$, the central charge $c$
 and the type $\psi_r$ if it is generated by a
Whittaker vector
of the same pair, central charge
 and type.

\noindent The order and the rank of the Whittaker vector or module of a type $\psi_r$ are defined as  $r$
and ${\rm rank} \,\psi_r$, respectively.
\end{DD}
\noindent
Let us note that the standard notions of the Whittaker vectors and modules for the Virasoro algebra
as defined
 in \cite{Ondrus09,Ondrus11}
correspond
to the case $r=1$ with no central charge condition.

We assume that all Whittaker vectors and modules in this section
are of the Whittaker  pairs ${\cal V}_{r}\subset {\cal V}$ and have the same fixed value of the central charge.
For a Whittaker module $V$ of a Whittaker pair ${\cal V}_{r}\subset {\cal V}$ and a type $\psi_r$ we shall
use compact notation $V_{\psi_r}$.
\begin{DD}
 A Whittaker module $W_{\psi_r}$ generated by $\ket{w}$ is called a universal Whittaker  module of
 a type $\psi_r$, if for any other Whittaker
 module $V_{\psi_r}$ generated by $\ket{v}$ there exists a surjective module homomorphism
 $\Phi: W_{\psi_r} \rightarrow V_{\psi_r}$ such that $\Phi\ket{w} = \ket{v}$.
The Whittaker vector $\ket{w}$ of the type $\psi_r$ generating the universal module $W_{\psi_r}$ is called the universal Whittaker vector of type $\psi_r$.
\end{DD}

\begin{TT}
For each
Lie algebra homomorphism
$\psi_r : \mathcal{V}_r \rightarrow \mathbb{C} $
there exists  a unique, up to an isomorphism, universal Whittaker module
$W_{\psi_r} $.
\end{TT}

\begin{proof}
For each
Lie algebra homomorphism
$\psi_r : \mathcal{V}_r \rightarrow \mathbb{C} $
\[
 I_r \equiv \sum_{i \geq r} U(\mathcal{V})(L_i - \psi_r (L_i))+U(\mathcal{V})( z-c)
\]
is a left ideal in
$U(\mathcal{V})$.
$U(\mathcal{V})/I_r$ with the left action
\[
 u[\,v\,] \equiv [\,uv\,], \hspace{10pt} u,v \in U(\mathcal{V})
\]
is a $U(\mathcal{V})$-module generated by
$[\,1\,]$. For any
$ n \geqslant r$ one has
\[
(L_n - \psi_r(L_n))[\,1\,] = [L_n - \psi_r(L_n)] = 0\;\;{\rm and} \;\;(z-c)[\,1\,]=0.
\]
Hence
$[\,1\,]$ is a Whittaker vector and
$U(\mathcal{V})/I_r$ is a Whittaker module.
\newline
Let
$V_{\psi_r}$ be an arbitrary Whittaker module of a type
$\psi_r$ generated by
a Whittaker vector
$\ket{v}$. The map
$$
\Phi: U(\mathcal{V})/I_r
\,\ni\, [\,u\,] \;\rightarrow
\;u\ket{v}\,\in \,
V_{\psi_r }
$$
is a  surjective homomorphism
such that
$\Phi([\,1\,]) = \ket{v}$.
The uniqueness is a simple consequence of the universality.
\end{proof}

We define a pseudo partition $\lambda
$ of order $r$ as a non-decreasing finite sequence of integers smaller than $r$:
$$
\lambda = (\lambda_1,\dots, \lambda_n ),
\hspace{10pt}
\lambda_1 \leqslant \hdots \leqslant \lambda_n < r.
$$
With each  $\lambda$ we associate an element of the universal enveloping algebra $U(\mathcal{V})$:
$$
L_{\lambda} \equiv L_{\lambda_1} \hdots L_{\lambda_n}.
$$
It is convenient to
 supplement the set ${\cal P}^r$ of all pseudo partitions of order $r$ by the empty sequence $\emptyset$
for which
$$
L_{\emptyset}\equiv 1.
$$
Sometimes it is useful  to denote a pseudo partition $\lambda$ of order $r$ as:
$$
\lambda = (\lambda(-l),  \hdots, \lambda(r-1))
$$
where $\lambda(k)$ is the number of times the integer $k$ appears in  $\lambda$.
In this notation  $L_{\lambda}$ takes  the form
$$
L_{\lambda} = L_{-l}^{\lambda(-l)} \hdots L_{r-1}^{\lambda(r-1)}.
$$
As a consequence of the PBW theorem one has

\begin{TT}
 Let $W_{\psi_r}$ be the universal Whittaker module generated by $\ket{w}$.
Then the vectors
$L_{\lambda}\ket{w}$ where $\lambda$ runs over the set ${\cal P}^r$ of all pseudo  partitions
of order $r$, form a basis of $W_{\psi_r}$.
\end{TT}

\noindent Each pseudo partition $\lambda$ can be uniquely decomposed
$
\lambda = \lambda_- \cup \lambda_+
$
into the negative
$
\lambda_- = (\lambda(-l),  \hdots, \lambda(-1))
$ and the nonnegative $
\lambda_+ = (\lambda(0),  \hdots, \lambda(r-1))
$
part.
We shall introduce the length  $\# \lambda $
$$
\# \lambda = \sum_{i=0}^{r-1} \lambda(i)
$$
and the level $|\lambda|$
$$
|\lambda| = - \sum_{i<0} i \lambda(i)
$$
of a pseudo partition $\lambda$.
For an arbitrary vector
$\ket{v} = \sum_{\lambda}p_{\lambda}L_{\lambda}\ket{w}$ we define the maximal length
$$
\text{max}^{\#}\ket{v}=\text{max}\{  \#\lambda \hspace{5pt} : \hspace{5pt} p_{\lambda} \neq 0\}
$$
and the maximal level
$$
\text{max}\,\ket{v}
= \text{max}\{  |\lambda|\hspace{5pt} : \hspace{5pt} p_{\lambda} \neq 0\}$$
\begin{LL}
\label{lemat_plus}
Let $\ket{w}$ be the universal Whittaker vector of a type $\psi_r$ and a rank $s$. Then
for any pseudo partition $\lambda\in {\cal P}^r$
$$
\begin{array}{rlllllllll}
   [L_m,L_{\lambda_{+}}]\ket{w} &=& 0 & &\text{for} & s<m     ,
   \\
   {\rm max}^{\#}[L_m,L_{\lambda_{+}}]\ket{w}
   &< & \#\lambda
   & &\text{for} &
    r\leqslant m \leqslant s         .
   \end{array}
$$
\end{LL}

\begin{proof}
 For $L_{\lambda_{+}} = L_{\lambda_{k+1}} \hdots L_{\lambda_n}$ one gets
 $$
 [L_m,L_{\lambda_{+}}]
 =
 \sum^{n}_{i=k+1}(m - \lambda_i) L_{\lambda_{k+1}} \hdots L_{m + \lambda_i} \hdots L_{\lambda_n}.
 $$
 If $m>s$ then $m+\lambda_i > s$ for all $i$ in the sum. Commuting all generators $L_l$ with $l>s$
 to the right one gets only  terms annihilating
 $\ket{w}$.
Terms with a maximal number of $L$ generators are
of the form
$$
L_{\lambda_{k+1}} \hdots L_{\lambda_{i-1}}L_{\lambda_{i+1}} \hdots L_{\lambda_n}L_{m + \lambda_i}\ket{w}
=
\psi_r(L_{m + \lambda_i})
L_{\lambda_{k+1}} \hdots L_{\lambda_{i-1}}L_{\lambda_{i+1}} \hdots L_{\lambda_n}\ket{w}.
$$
They do not necessarily vanish for
$r\leqslant m \leqslant s$. Hence the maximal possible number of generators is   $\#\lambda-1$.
\end{proof}

\begin{LL}
\label{lemat_minus}
Let $\ket{w}$ be the universal Whittaker vector of a type $\psi_r$ and a rank $s$. Then
for any pseudo partition $\lambda\in {\cal P}^r$
$$
\begin{array}{rlllllllll}
   \big[L_m,L_{\lambda_{-}}\big]L_{\lambda_{+}}\ket{w} &=& 0
   & & \text{for} & s+|\lambda|<m     ,
   \\   [5pt]
   {\rm max}\big[L_m,L_{\lambda_{-}}\big]L_{\lambda_{+}}\ket{w}
   &\leqslant & |\lambda|+s-m
   & &\text{for} &
    s < m \leqslant s+|\lambda|      ,
   \\    [7pt]
   {\rm max}\big[L_m,L_{\lambda_{-}}\big]L_{\lambda_{+}}\ket{w}
   &<& |\lambda|
   & &\text{for} &
   r\leqslant m \leqslant s   .
   \end{array}
$$
\end{LL}

 \begin{proof}
 Let $s+|\lambda|<m $.  Ordering all terms in  the commutator $[L_m, L_{\lambda_{-}}]$
 one gets a sum of monomials with the generators $L_l$
   on the right such that $l>s$.  Hence  $[L_n,L_{\lambda_{-}}]L_{\lambda_{+}}\ket{w} = 0$  by Lemma \ref{lemat_plus}.
 \newline
 Let $r\leqslant m \leqslant s+|\lambda|$. For $L_{\lambda_{-}} = L_{\lambda_1} \hdots L_{\lambda_k}$
 one has
 \begin{eqnarray*}
 [L_m,L_{\lambda_{-}}] L_{\lambda_{+}}\ket{w}
 &=&
 \sum_{m+\lambda_i<0}(m - \lambda_i) L_{\lambda_1} \hdots L_{m + \lambda_i} \hdots L_{\lambda_k}L_{\lambda_{+}}\ket{w}
 \\
 &+&
 \sum_{m+\lambda_i\geqslant 0}(m - \lambda_i)
 L_{\lambda_1} \hdots L_{m + \lambda_i} \hdots L_{\lambda_k}L_{\lambda_{+}}\ket{w}.
 \end{eqnarray*}
 After reordering the first sum  takes the form
 $$
 \underset{ \scriptstyle \# \gamma = \,0}{\sum_{ |\gamma|\,=\,|\lambda|-m }}
\hspace{-12pt}p_{\gamma} \,L_{\gamma}\,L_{\lambda_{+}}\,\ket{w}
  $$
  and if nonzero it is of level $|\lambda|-m $.
 The second sum can be rewritten as
 $$
 \underset{ \scriptstyle \# \gamma = \,0}{\sum_{ |\gamma|
 \,=\,
 |\lambda|-m }}
 \hspace{-8pt}p_{\gamma} \,L_{\gamma}\,L_{\lambda_{+}}\,\ket{w}
 +\sum\limits_{l}^{}
 \underset{ \scriptstyle \# \gamma = \,0}{\sum_{ |\gamma|\,=\,|\lambda|-l}}
 \hspace{-6pt}p_{\gamma} \,L_{\gamma}\,L_{m-l}\,L_{\lambda_{+}}\,\ket{w},
 $$
 where the sum over $l$ runs over the partial sums
 $$
 m> l=-\sum\limits_{j=1}^t \lambda_{i_j}\geqslant m-s
 $$
corresponding to all possible subseries $\{\lambda_{i_1},\dots,\lambda_{i_t}\}
\subset \{\lambda_{1},\dots,\lambda_{k}\}$. Let $l'$ be the smallest of  such sums. Then the
terms
 $L_{\gamma}\,L_{m-l'}\,L_{\lambda_{+}}\,\ket{w}$ are of the maximal possible level
 $$
 {\rm max}\, L_{\gamma}\,L_{m-l'}\,L_{\lambda_{+}} \, \ket{w}
 =
 |\lambda| - l' \leqslant |\lambda| +s -m.
 $$
For $r \leqslant m \leqslant s $ one simply has
$$
{\rm max}\, L_{\gamma}\,L_{m-l'}\,L_{\lambda_{+}} \, \ket{w}
 =
 |\lambda| - l' < |\lambda| .
 $$
 \end{proof}

\begin{LL}
\label{lemat_kplus2r}
Let $\ket{w}$ be the universal Whittaker vector of a type $\psi_r$ and a rank $s$ and $\lambda$ - a pseudo partition
of order $r$.
If $k>0$ is the smallest number for which
$\lambda(-k) \neq 0$, then
\begin{eqnarray*}
[L_{k+s},L_{\lambda_{-}}]L_{\lambda_{+}}\ket{w}
&=&\\
&&
\hspace{-70pt}
\lambda(-k) \psi_r(L_{s})(2k+s)
L_{-l}^{\lambda(-l)} \hdots L_{-k-1}^{\lambda(-k-1)} L_{-k}^{\lambda(-k)-1} L_{\lambda_{+}}\ket{w}
\; + \;\ket{v}\; + \textstyle\;\ket{v'},
\end{eqnarray*}
where ${\rm max}\,\ket{v}  < |\lambda| - k$ and ${\rm max}\,\ket{v'}  \leqslant |\lambda| - k,\; {\rm max}^\#\ket{v'}  < \#\lambda$.
\end{LL}

\begin{proof}
Let $L_{\lambda_{-}^{'}} = L_{-n}^{\lambda(-n)} \hdots L_{-k-1}^{\lambda(-k-1)}$, so that
$L_{\lambda_{-}} = L_{\lambda_{-}^{'}}L_{-k}^{\lambda(-k)}$. Then
\begin{equation}
\label{komutator}
[L_{k+s},L_{\lambda_{-}}]L_{\lambda_{+}}\ket{w}
= [L_{k+s},L_{\lambda_{-}^{'}}]L_{-k}^{\lambda(-k)}L_{\lambda_{+}}\ket{w}
+ L_{\lambda_{-}^{'}}[L_{k+s},L_{-k}^{\lambda(-k)}]L_{\lambda_{+}}\ket{w}.
\end{equation}
Repeating the reasoning from the proof of Lemma
\ref{lemat_minus}
one can write
\begin{eqnarray*}
[L_{k+s},L_{\lambda_{-}^{'}}]L_{-k}^{\lambda(-k)}L_{\lambda_{+}}\ket{w}
&=&
\hspace{-20pt}
\underset{ \scriptstyle \# \gamma = \,0}{\sum_{ |\gamma|\,=\,|\lambda'|-k-s }}
\hspace{-10pt}p_{\gamma} \,L_{\gamma}\, L_{-k}^{\lambda(-k)}L_{\lambda_{+}}\ket{w}
\\
&& \hspace{-50pt}
+
\sum\limits_{l=\,k+1}^{k+s}
\underset{ \scriptstyle \# \gamma = \,0}{\sum_{ |\gamma|\,=\,|\lambda'|-l}}
\hspace{-10pt}p_{\gamma} \,L_{\gamma}\,L_{k+s-l} L_{-k}^{\lambda(-k)}L_{\lambda_{+}}\ket{w}
\end{eqnarray*}
where the range in the sum over $l$ follows from the assumption that
$L_{\lambda_{-}^{'}}$ does not contain generators $L_{-i}$ with $i\leqslant k$.
An element of the maximal degree in the second sum
takes the form
$$
q_{\gamma}
 L_{\gamma} L_{-k}^{\lambda(-k)}L_{s-1}L_{\lambda_{+}}\ket{w},
$$
where $|\gamma| =|\lambda'|-k-1$. Hence
$$
{\rm max}\,[L_{k+s},L_{\lambda_{-}^{'}}]L_{-k}^{\lambda(-k)}L_{\lambda_{+}}\ket{w}
< |\lambda|-k.
$$
Let us now turn to the second term in
\eqref{komutator}. A simple algebra yields
\begin{eqnarray*}
L_{\lambda_{-}^{'}}[L_{k+s},L_{-k}^{\lambda(-k)}]L_{\lambda_{+}}\ket{w}
&=& (2k+s)L_{\lambda_{-}^{'}}
  \sum_{j=1}^{\lambda(-k)}L_{-k}^{j-1}
  L_{s}L_{-k}^{\lambda(-k)-j}L_{\lambda_{+}}\ket{w}
\\
&=& \lambda(-k)\psi_r(L_s) (2k+s)
L_{\lambda_{-}^{'}} L_{-k}^{\lambda(-k)-1}L_{\lambda_{+}}\ket{w}
\\
 &+&
(2k+s)L_{\lambda_{-}^{'}}
  \sum_{j=1}^{\lambda(-k)}L_{-k}^{j-1}
  [L_{s},L_{-k}^{\lambda(-k)-j}]L_{\lambda_{+}}\ket{w}
  \\
&+&
\lambda(-k)(2k+s)
L_{\lambda_{-}^{'}} L_{-k}^{\lambda(-k)-1}[L_s,L_{\lambda_{+}}]\ket{w}.
\end{eqnarray*}
By Lemma
\ref{lemat_minus}
maximal level of the second term on the r.h.s. is strictly smaller than $|\lambda| -k$.
If $\lambda_+(0)\neq 0$  the last term does not vanish. Since
\[
[L_s, L_{\lambda_+}]|w\rangle= [L_s, L_0^k]L_{\lambda'_+}|w\rangle = \psi_{r}(L_s)\sum_{l=1}^{k}\binom{k}{l}s^l L_{0}^{k-l}L_{\lambda'_+}|w\rangle
\]
its length is strictly smaller than $\# \lambda_+$.
\end{proof}
\begin{LL}
\label{lemat_no Whittaker vectors}
Let $W_{\psi_r}$ be the universal Whittaker module of type $\psi_r$ and let
$\ket{u}\in W_{\psi_r}$ be an arbitrary vector.
If ${\rm max}\,\ket{u} >0$ then $\ket{u}$ is not a Whittaker vector of any pair ${\cal V}_{r'}\subset {\cal V}$.
\end{LL}

\begin{proof}
 For an arbitrary nonzero vector $\ket{u} = \sum p_{\lambda} L_{\lambda}\ket{w} $
 we introduce
 $$
 M=\max \, \ket{u} ,
 \hspace{10pt}
 \Lambda_M = \{\lambda \hspace{5pt} : \hspace{5pt} p_{\lambda} \neq 0 \wedge |\lambda| = M\}.
 $$
Since $M>0$,  there exists   a smallest positive number $k$ for which there exists a partition
 $\lambda \in \Lambda_M$ such that $\lambda(-k) \neq 0$.
 For $s={\rm rank}\, \psi_r$ one has
 \begin{eqnarray*}
L_{k+s} \ket{u}
 &=& \sum_{\lambda \notin \Lambda_M}p_{\lambda}[L_{k+s},L_{\lambda_{-}}]L_{\lambda_{+}}\ket{w}
 + \sum_{\lambda \in \Lambda_M}p_{\lambda}[L_{k+s},L_{\lambda_{-}}]L_{\lambda_{+}}\ket{w}.
 \end{eqnarray*}
If $\lambda \notin \Lambda_M$ then  $|\lambda|<M$ and Lemma
\ref{lemat_minus} implies
 $$
 {\rm max}
 \sum_{\lambda \notin \Lambda_M}p_{\lambda}[L_{k+s},L_{\lambda_{-}}]L_{\lambda_{+}}\ket{w}
 = |\lambda_{-}| -k < M - k.
 $$
Taking this into account and applying Lemma
\ref{lemat_kplus2r}
to all terms of the second sum one can write
\begin{equation}
\label{expr}
L_{k+s}\ket{u}
= \underset{\lambda(-k) \neq 0}{\sum_{\lambda \in \Lambda_M}}
p_{\lambda}\lambda(-k)\psi_r(L_{s})
  (2k+s)L_{\lambda'_{-}}L_{\lambda_{+}}\ket{w}
  \;+\;\ket{v}\;+\;\textstyle\ket{v'},
\end{equation}
where ${\rm max}\,\ket{v} < M-k$, ${\rm max}^\#\,\ket{v'} < \#\lambda $ and $\lambda'_{-}=(\lambda(-l),\dots, \lambda(-k-1), \lambda(-k)-1)$.
All vectors
 $L_{\lambda'_{-}}L_{\lambda_{+}}\ket{w}$ in the sum above are linearly independent.
Since for all of them $|L_{\lambda'_{-}}L_{\lambda_{+}}\ket{w}|= M-k$,
 they form with
$\ket{v}$ and $\ket{v'}$ a linearly independent system as well. The decomposition (\ref{expr}) thus
imply that $L_{k+s}\ket{u}$ is a nonzero vector  not proportional to  $\ket{u}$.
Hence  $\ket{u}$ is not a Whittaker vector of any type of order $k+s$ or lower.

On the other hand, for $n > k+s$ one has
$$
L_{n} \ket{u}
= \sum_{\lambda}p_{\lambda}L_n L_{\lambda_{-}}L_{\lambda_{+}}\ket{w}
= \sum_{\lambda}p_{\lambda}[L_n, L_{\lambda_{-}}]L_{\lambda_{+}}\ket{w}.
$$
Lemma
\ref{lemat_minus} now implies
$$
{\rm max}
\left(
      \sum_{\lambda}p_{\lambda}[L_n, L_{\lambda_{-}}]L_{\lambda_{+}}\ket{w}
\right)
 < {\rm max}\ket{u}
$$
so $L_n \ket{u}$ cannot be of the form $\alpha \ket{u}$ for any $\alpha \neq 0$.
Hence $\ket{u}$ is not a Whittaker vector of any order higher than $k+s$.
\end{proof}
For an arbitrary element $L_{\lambda_r}\ket{w}$ of the basis in the universal
Whittaker module $W_{\psi_r}$ with $\#\lambda\neq 0$  we denote by  $l_{\lambda}$  the smallest nonnegative integer
$i$ such that $\lambda(i)\neq 0$. It is convenient to assume that $l_{\lambda}=r$ if $\#\lambda=0$.
One has for instance $l_{\emptyset}=r$ for $\ket{w}$.
\begin{TT}
\label{tw_wektor}
Let $\ket{w}$ be a universal Whittaker vector of a type $\psi_r$.
All  Whittaker vectors  in $W_{\psi_r}$  of a given type $\psi_{r'}$ form a linear subspace
$Wh_{\psi_{r'}}\subset W_{\psi_r}$.
\begin{enumerate}
\item
If ${\rm rank}\,\psi_r =s\in \{2r,2r-1\} $ there are Whittaker vectors in $W_{\psi_r}$ of
the type $\psi_r$ and of the higher
order types
$$
\psi_{r'} = \{ \psi_r(L_{r'}),\dots,\psi_r(L_{s}), 0,\dots,0\}\;\;\;,\;\;\;r' = s- r+2, \dots, s,
$$
and
\vspace{-25pt}
\begin{eqnarray*}
Wh_{\psi_r}&=&{\rm span} \{ \ket{w}\},
\\
Wh_{\psi_{r'}}&=&{\rm span} \{L_{\lambda} \ket{w}\,:\,
|\lambda_r|=0, \,l_{\lambda} \geqslant s-r'+1\}.
\end{eqnarray*}
\item
If ${\rm rank}\, \psi_r\, =\,s<2r-1$
there are Whittaker vectors in $W_{\psi_r}$ of  the type $\psi_r$ and of the higher
order types
$$
\psi_{r'} = \{ \psi_r(L_{r'}),\dots,\psi_r(L_{s}), 0,\dots,0\}\;\;\;,\;\;\;r' =  r+1, \dots, s,
$$
and\vspace{-25pt}
\begin{eqnarray*}
Wh_{\psi_{r}}&=&{\rm span} \{L_{\lambda} \ket{w}\,:\,
|\lambda_r|=0,\, l_{\lambda} \geqslant s-r+1\},
\\
Wh_{\psi_{r'}}&=&{\rm span} \{L_{\lambda} \ket{w}\,:\,
|\lambda_r|=0,\, l_{\lambda} \geqslant s-r'+1\}.
\end{eqnarray*}
\end{enumerate}
There are no other Whittaker vectors of any type in the universal Whittaker module $W_{\psi_r}$.
\end{TT}

\begin{proof}
By Lemma \ref{lemat_no Whittaker vectors} it is enough to consider vectors of the form
$$
\ket{u}=\sum p_{\lambda} L_{\lambda}\ket{w}
$$
where $|\lambda|=0$ for all $\lambda$ in the sum.
Let $N =
{\rm max}^\#\ket{u}$ and
\begin{eqnarray*}
\Lambda^{N} &=&\{\lambda \,:\, p_{\lambda} \neq 0 \wedge \#\lambda = N\},
\\
l
&=&
{\rm min}\, \{ l_{\lambda} \,:\, \lambda\in\Lambda^{N}\},
\\
\Lambda^N_{l}&=&\{\lambda\in \Lambda^N \,:\, l_{\lambda}=l\}\ .
\end{eqnarray*}
For $k\geqslant r$ one has
 $$
 (L_{k}-\psi_r({L_k})) \ket{u}
 = \sum_{\lambda \notin \Lambda^{N}}
 p_{\lambda}[L_{k},L_{\lambda_{+}}]\ket{w}
 + \sum_{\lambda \in \Lambda^{N}}
 p_{\lambda}[L_{k},L_{\lambda_{+}}]\ket{w}.
 $$
For the first sum Lemma
\ref{lemat_plus} implies
$$
{\rm max}^\# \sum_{\lambda \notin \Lambda^{N}}
 p_{\lambda}[L_{k},L_{\lambda_{+}}]\ket{w}
< N-1.
$$
If $s-l\geqslant r$ the second sum for $k=s-l$ takes the form
\begin{eqnarray*}
\sum_{\lambda \in \Lambda^{N}}
 p_{\lambda}[L_{s-l},L_{\lambda_{+}}]\ket{w}
&=&
\sum_{\lambda \in \Lambda^{N}_{l}}
 p_{\lambda}
\psi_r(L_{s})\lambda(l) (s-2l)L_{l}^{\lambda(l)-1}  \hdots  L_{r-1}^{\lambda(r-1)} \ket{w}
\end{eqnarray*}
and one one gets
\begin{eqnarray*}
(L_{s-l}- \psi_{r}(L_{s-l}))\ket{u} &=& \ket{v}
+
\sum_{\lambda \in \Lambda^{N}_l}
 p_{\lambda}
\psi_r(L_{s}) \lambda(l) (s-2l) L_{l}^{\lambda(l)-1}  \hdots  L_{r-1}^{\lambda(r-1)} \ket{w}
\end{eqnarray*}
where $\text{max}^{\#}\,\ket{v} < N-1$
so that all the terms on the r.h.s are linearly independent.
It follows
that $L_{s-l} \ket{u} $ is a nonzero vector not proportional to $\ket{u}$. Hence $\ket{w'}$ is not a Whittaker vector
of any type of order $s-l$ or lower.

\noindent We shall now discuss separate cases.
\begin{enumerate}
\item
If $s \in \{2r-1,2r\}$ the condition $s-l\geqslant r$ is always satisfied and there are no Whittaker vectors in $W_{\psi_r}$ of
any type of the order $r$ or lower except vectors proportional to $
\ket{w}$.
For $r'$ in the range $s-r+2 \leqslant r' \leqslant s$ all vectors of the form
$$
\underset{l_{\lambda} \geqslant s-r'+1}{\sum\limits_{|\lambda|=0}} p_{\lambda} L_{\lambda} \ket{w}
$$
are Whittaker vectors of  type
$
\psi_{r'} = \{ \psi_r(L_{r'}),\dots,\psi_r(L_{s}), 0,\dots,0\}.
$
\item
If $s < 2r-1$  there are no Whittaker vectors of any type of order $s$ or lower
with nonzero components along the vectors
$L_{\lambda}\ket{w}$ with $s- l_{\lambda}\geqslant r$ ($|\lambda|=0$). The only possibility left over are vectors of the form
$$
\underset{l_{\lambda}>s-r}{\sum\limits_{|\lambda|=0}} p_{\lambda} L_{\lambda} \ket{w}.
$$
One easily checks that all of them are Whittaker vectors of the same type as $\ket{w}$.
Thus there are no lower order Whittaker vectors of any type and all Whittaker vectors of order $r$ are of the $\psi_r$ type.

The higher order vectors can be constructed as in the other two cases.
For $r'$ in the range $r+1\leqslant r' \leqslant s$ all vectors of the form
$$
\underset{l_{\lambda} \geqslant s-r'+1}{\sum\limits_{|\lambda|=0}} p_{\lambda} L_{\lambda} \ket{w}
$$
are Whittaker vectors of the type
$
\psi_{r'} = \{ \psi_r(L_{r'
}),\dots,\psi_r(L_{s}),0,\dots, 0\}.
$
\end{enumerate}
Since all vectors $\ket{u}$ of the maximal level zero were considered there are no other Whittaker vectors.
\end{proof}

\begin{DD}
Let  $V_{\psi_r}$ be a Whittaker module generated by a Whittaker vector $\ket{w}$ of a type $\psi_r$.
The dot-action $\cdot$ of the subalgebra ${\cal V}_r$ on $V_{\psi_r}$ is defined
by
$$
L_n \cdot \ket{v} \equiv (L_n - \psi_r(L_n))\ket{v},\;\;\;n\geqslant r,\;\;\;\ket{v} \in V_{\psi_r}.
$$
\end{DD}
\noindent
For an arbitrary vector $\ket{v} = \sum p_{\lambda_r}L_{\lambda_r}\ket{w}$ one has
$$
L_n \cdot \ket{v} = \sum p_{\lambda_r}[L_n,L_{\lambda_r}]\ket{w}.
$$
\begin{LL}
\label{lemat_nilpotent}
All generators of $\,{\cal V}_r$ are locally nilpotent on  $V_{\psi_r}$ with respect to the dot-action,
i.e. for each $n \geqslant r$ and  $\ket{v} \in V_{\psi_r}$ there exists an integer $k_{n,\ket{v}}$
such that
$$
L^{k_{n,\ket{v}}}_n \ket{v}=0.
$$
\end{LL}

 \begin{proof}
 It is enough to consider vectors of the form $L_{\lambda}\ket{w}$.
  For any $n \geq r$ and $\lambda \in \mathcal{P}^r$
  one has, up to numerical coefficients,
 \begin{eqnarray}
 \label{lemat_nilpotent_komutator}
   \underset{k}{\underbrace{[L_n,[L_n,\hdots[L_n}},L_{\lambda_{-}}L_{\lambda_{+}}]\hdots]\ket{w}
 &\sim&
 \\
 &&
 \hspace{-150pt} \sim \sum_{l=0}^{k}
   \underset{l}{\underbrace{[L_n,[L_n,\hdots[L_n}},L_{\lambda_{-}}]\hdots]
   \underset{k-l}{\underbrace{[L_n,[L_n,\hdots[L_n}},L_{\lambda_{+}}]\hdots]\ket{w}.
 \nonumber
\end{eqnarray}
 By Lemmas
\ref{lemat_plus}
and
\ref{lemat_minus}
the level
of the expression $[L_n, L_{\lambda_{-}}]L_{\lambda_{+}}\ket{w}$ is smaller than
$|\lambda|$ and the length of the expression
$L_{\lambda_{-}}[L_n, L_{\lambda_{+}}]\ket{w}$ is smaller than $\#\lambda$.
Let $k_+$ and $k_-$ be the biggest numbers for which
\begin{eqnarray*}
{\rm max}\, \underset{k_-}{\underbrace{[L_n,[L_n,\hdots[L_n}},L_{\lambda_{-}}]\hdots]\ket{w}&>&0,
 \\
{\rm max}^\#\,\underset{k_+}{\underbrace{[L_n,[L_n,\hdots[L_n}},L_{\lambda_{+}}]\hdots]\ket{w}&>&0.
\end{eqnarray*}
To ensure vanishing of
\eqref{lemat_nilpotent_komutator} one can  choose $
k> 2\max\{k_+ +1,k_- +1\}.
$
 \end{proof}

\begin{LL}
\label{lemat_finite_dim}
Let $\ket{v} \in V_{\psi_r}$ be an arbitrary  vector.
$U(\mathcal{V}_r)\cdot \ket{v}$ is a finite
dimensional $\mathcal{V}_r$ submodule of $V_{\psi_r}$ with respect to the dot action.
\end{LL}
\begin{proof}
The PBW basis of $U(\mathcal{V}_r)$ consists of all monomials of the form
$
L_r^{\alpha(r)}\dots L_n^{\alpha(n)}
$.
By Lemmas
\ref{lemat_plus} and
\ref{lemat_minus}
there exists $N$ such that $L_n \cdot \ket{v}=0$
for all $n> N$.
Using Lemma \ref{lemat_nilpotent} one can then show that there are only finitely many pseudo partitions
$$
\lambda^N =( \lambda^N(r), \dots, \lambda^N(N-1))
$$
such that $L_{\lambda^N}\cdot\ket{v}\neq 0$.
\end{proof}
\begin{TT}
 Any submodule of a Whittaker module of a type $\psi_r$ contains a Whittaker vector of the same type.
\label{tw_podmodul}
\end{TT}

\begin{proof}
 Let $S$ be a submodule of a Whittaker module $V_{\psi_r}$ and $0 \neq \ket{v} \in S$.
By Lemma \ref{lemat_finite_dim} $F=U(\mathcal{V}_r)\cdot \ket{v}$ is a finite
${\cal V}_r$ submodule with respect to the dot action.
By Lemmas
\ref{lemat_plus} and
\ref{lemat_minus}
there exists $N$ such that $L_n \cdot F=0$
for all $n> N$.
The quotient ${\cal U}={\cal V}_r / {\cal V}_N$ is a finite dimensional Lie algebra
and $F$ is a ${\cal U}$ module with respect to the induced dot action
$$
[L_n]\cdot \ket{u} = L_n\cdot \ket{u},\;\;\;n=r,\dots,N-1,\;\;\ket{u}\in F.
$$
It follows from Lemma \ref{lemat_nilpotent}  that this action is locally nilpotent.
Thus  by the Engel theorem  there exists a vector $\ket{w'}\in F$ such that
$$
[L_n]\cdot \ket{w'} = 0,\; \;\;n=r,\dots,N-1.
$$
But this is a Whittaker vector of the type $\psi_r$ by construction of the induced dot action.
\end{proof}

\begin{TT}
\label{tw_structure of universal}
 Let $W_{\psi_r}$ be the universal Whittaker module generated by $\ket{w}$ and $s={\rm rank}\,\psi$.
 \begin{enumerate}
\item
If $s \in \{2r,2r-1\}$ then $W_{\psi_r}$ is simple.
\item
If $s<2r-1$ then there exists an infinite composition series
$$
\dots W^{(n)}_{\psi_r} \,\subset\,\dots\,\subset\,W^{(1)}_{\psi_r}\,\subset\, W^{(0)}_{\psi_r}=W_{\psi_r},
$$
such that all composition factors
$W^{(n-1)}_{\psi_r} /  W^{(n)}_{\psi_r}
$
are Whittaker modules of the type
$$
\begin{array}{lllllll}
\psi_{{s\over 2}} &=& \{ 0,\dots,0,\psi_{r}(L_{r}),\dots .\,.\,,\psi_r(L_{s})\}
&{\it for}& s\in 2\mathbb{N},
\\
\psi_{{s-1\over 2}} \!&=& \{ 0,\dots,0,\psi_{r}(L_{r}),\dots,\psi_r(L_{s}), 0\}
&{\it for}& s\in 2\mathbb{N}+1.
\end{array}
$$
\end{enumerate}
\end{TT}

\begin{proof}
If $S \subset W_{\psi_r}$ is a submodule then
 Theorem \ref{tw_podmodul} implies that $S$ contains  a Whittaker vector $\ket{w'}$ and also the submodule
 $U({\cal V})\ket{w'}$ generated by $\ket{w'}$. By Theorem \ref{tw_wektor}, if $s\in \{2r,2r-1\}$ then any
 Whittaker vector of the type $\psi_r$ in $W_{\psi_r}$ is proportional to $\ket{w}$ hence $U({\cal V})\ket{w'}=W_{\psi_r}$
 and $S=W_{\psi_r}$ which proves the first part.

For the proof of the second part we
construct a strictly decreasing series of submodules:
$$
W^{(n)}_{\psi_r} ={\rm span}_{U({\cal V})} \{
L_{\scriptscriptstyle\left[{s\over 2}\right]} \ket{w},\dots,
L_{r-2}\ket{w}, L^n_{r-1}\ket{w}
\}
$$
For all $n\in\mathbb{N}$ the quotient $W^{(n-1)}_{\psi_r} /  W^{(n)}_{\psi_r}$
 is generated by the vector $\left[\,L^n_{r-1}\ket{w}\,\right]$
which by construction is a Whittaker vector of one of the types stated above. By the first part of the theorem
the quotient module is simple. The construction is not unique. It works for subsequent powers of an arbitrary
linear combinations of the form
$$
\alpha^{\scriptscriptstyle\left[{s\over 2}\right]}  L_{\scriptscriptstyle\left[{s\over 2}\right]}
+
\dots
+
\alpha^{r-1}  L_{r-1}.
$$
\end{proof}

\begin{CC}
\label{wniosek 1}
Any Whittaker module of order $r$ and rank $2r$ or $2r-1$ is isomorphic to the universal Whittaker module of the same type.
\end{CC}
\begin{proof}
Let $V_{\psi_r}$ be a Whittaker module of a type $\psi_r$ and rank $2r$ or $2r-1$ generated by a vector $\ket{v}$.
By the universal property of $W_{\psi_r}$ there exists a surjective homomorphism
$\Phi: W_{\psi_r}\to V_{\psi_r}$ such that $\Phi(\ket{w})=\ket{v}$. The kernel of $\Phi$ is a submodule of
$W_{\psi_r}$. By Theorem \ref{tw_structure of universal} $W_{\psi_r}$ is simple and $\Phi(\ket{w})=\ket{v}\neq 0$
hence ${\rm ker}\,\Phi=\{\ket{\,0\,}\}$. Thus $\Phi$ is an isomorphism of ${\cal V}$ modules.
\end{proof}
As an immediate consequence of Corollary \ref{wniosek 1} and Theorem \ref{tw_structure of universal} one has
\begin{CC}
Any Whittaker module of  order $r$ and  rank $2r$ or $2r-1$ is simple.
\end{CC}

\section{Whittaker pairs ${\cal V}_{1,n}\subset {\cal V}$ }

Let $n$ be a  positive integer. We introduce the subalgebra
$\mathcal{V}_{1,n}$
generated by $L_{1}, L_{n}$
 $$
 \mathcal{V}_{1,n}
  = \text{span}_{\mathbb{C}}\{L_1, L_{n}, L_{n+1},  \hdots \}\subset {\cal V}.
 $$
A
Lie algebra homomorphism
$\psi_{1,n} : \mathcal{V}_{1,n} \rightarrow \mathbb{C} $
is uniquely defined by its values on the algebra generators.
As it was mentioned in the introduction the cases when $\psi_{1,n}(L_1)=0$ or $\psi_{1,n}(L_n)=0$
are already described in the previous section. So is the case $n=2$.
We shall assume therefore that homomorphism $\psi_{1,n}$ is regular i.e.
$\psi_{1,n}(L_1)\neq 0$, $\psi_{1,n}(L_n)\neq 0$ and $n>2$.
From commutation relations (\ref{virasoro}) one has $\psi_{1,n}(L_k) = 0$ for all $k>n$.
\begin{DD}
Let $V$ be a left $\mathcal{V}$-module,  $\psi_{1,n} : \mathcal{V}_{1,n} \rightarrow \mathbb{C} $
- a non trivial  Lie algebra homomorphism and $c$ - a complex number. A vector $\ket{w} \in V$ is called a Whittaker vector of
the Whittaker pair ${\cal V}_{1,n}\subset {\cal V}$,
the central charge $c$ and the type $\psi_{1,n}$
if
$$
z\ket{w} =c\ket{w},\;\;L_1 \ket{w} = \psi_{1,n}(L_1) \ket{w}
\;\;{\it and }\;\;L_k \ket{w} = \psi_{1,n}(L_k) \ket{w}\;\; {\it for} \;\;k\geqslant n.
$$
A $\mathcal{V}$-module $V$ is called a Whittaker module of the Whittaker pair ${\cal V}_{1,n}\subset {\cal V}$, the central charge $c$
 and the type $\psi_{1,n}$ if it is generated by a
Whittaker vector
of the same central charge
 and type.

\noindent We say that the Whittaker pair ${\cal V}_{1,n}\subset {\cal V}$ and its Whittaker vectors are of order  $n$.
\end{DD}
\noindent
For a Whittaker module $V$ of a Whittaker pair ${\cal V}_{1,n}\subset {\cal V}$ and a type $\psi_{1,n}$ we shall
use compact notation $V_{\psi_{1,n}}$.
As before we assume that all Whittaker vectors and modules in this section are of
Whittaker pairs ${\cal V}_{1,n}\subset {\cal V}$ and have the same fixed value of the central charge.

Both the definition of the universal Whittaker module
for the pair ${\cal V}_{1,n}\subset {\cal V}$
and  proofs of the theorems below are obvious modifications of the considerations of the previous section.

\begin{TT}
For each
Lie algebra homomorphism
$\psi_{1,n} : \mathcal{V}_{1,n} \rightarrow \mathbb{C} $
there exists  a unique, up to an isomorphism, universal Whittaker module
$W_{\psi_{1,n}} $.
\end{TT}

\begin{TT}
 Let $W_{\psi_{1,n}}$ be the universal Whittaker module generated by $\ket{w}$.
Then the vectors
$L_{\lambda}\ket{w}$ where $\lambda$ runs over the set
$$
{\cal P}^{1,n}=\{ \lambda \in {\cal P}^n \,:\, \lambda(1)=0\},
$$
form a basis of $W_{\psi_{1,n}}$.
\end{TT}
The counterpart of Lemma \ref{lemat_plus} takes the form
\begin{LL}
\label{lemat_plus_1}
Let $\ket{w}$ be the universal Whittaker vector of a type $\psi_{1,n}$. Then
for any pseudo partition $\lambda\in {\cal P}^{1,n}$
\begin{enumerate}
\item
$[L_m,L_{\lambda_{+}}]\ket{w} = 0$ for $ m>n $,
\item
$   {\rm max}^{\#}[L_n,L_{\lambda_{+}}]\ket{w}<\#\lambda$,
\item
there exists a positive integer $m_\lambda$ such that

$  {\rm max}^\#\,\underset{m_\lambda}{\underbrace{[L_1,[L_1,\hdots[L_1}},L_{\lambda_{+}}]\hdots]\ket{w}<\#\lambda$.
\end{enumerate}
\end{LL}
By strightforward extensions of Lemmas \ref{lemat_minus}, \ref{lemat_kplus2r} one gets

\begin{LL}
\label{lemat_no Whittaker vectors_1}
Let $W_{\psi_{1,n}}$ be the universal Whittaker module of type $\psi_{1,n}$ and let
$\ket{u}\in W_{\psi_{1,n}}$ be an arbitrary vector.
If ${\rm max}\,\ket{u} >0$ then $\ket{u}$ is not a Whittaker vector of any type.
\end{LL}

\begin{LL}
\label{lemat_no Whittaker vectors_2}
Let $W_{\psi_{1,n}}$ be the universal Whittaker module $ W_{\psi_{1,n}}$ of type $\psi_{1,n}$ and let
$$
\ket{u}=\sum p_{\lambda} L_{\lambda}\ket{w}
$$
be an arbitrary vector in $ W_{\psi_{1,n}}$.
If $p_\lambda\neq 0$ for pseudo-partitions $\lambda\in {\cal P}^{1,n}$ with $\lambda(0)\neq0$
then $\ket{u}$ is not a Whittaker vector of any type.
\end{LL}

\begin{proof}
By Lemma \ref{lemat_no Whittaker vectors_1} it is enough to consider vectors of the form
$$
\ket{u}=\sum p_{\lambda} L_{\lambda}\ket{w}
$$
where $|\lambda|=0$ for all $\lambda$ in the sum.
Let $\Lambda^0$ be
the set of all partitions
 such that  $p_\lambda\neq 0$ and $\lambda(0) \neq 0$.
Let
$$
\lambda^0_{\rm max}=
{\rm max} \{\lambda(0)\,:\, \lambda\in\Lambda^0\}.
$$
One can write
\begin{eqnarray*}
(L_n - \psi_{1,n}(L_n))\ket{u}
&=&
\underset{\lambda(0)=\lambda^0_{\rm max}}{\sum_{\lambda \in \Lambda^0 }}
 p_{\lambda} [L_n, L_{\lambda}]\ket{w}
+
\underset{\lambda(0)<\lambda^0_{\rm max}}{\sum_{\lambda \in \Lambda^0 }}
 p_{\lambda} [L_n, L_{\lambda}]\ket{w}
\\
&=&
n\lambda^0_{\rm max}\, \psi_{1,n}(L_n)\underset{\lambda(0)=\lambda^0_{\rm max}}{\sum_{\lambda \in \Lambda^0 }}
p_{\lambda} L_0^{\lambda(0)-1}L_{1}^{\lambda(1)}\dots L_{n-1}^{\lambda(n-1)} \ket{w}
\\
&+&
{\sum_{\lambda'(0)<\lambda^0_{\rm max}-1}}
 p'_{\lambda'}  L_{\lambda'}\ket{w}.
\end{eqnarray*}
All terms of the first sum with the whole second sum  form a set of linearly independent vectors.
Hence, as each term in the first sum does not vanish neither the whole sum does.
\end{proof}

Let us now turn to the analysis of Whittaker vectors in the universal Whittaker modules $ W_{\psi_{1,n}}$ of a type $\psi_{1,n}$.
By Lemma \ref{lemat_no Whittaker vectors_1} and \ref{lemat_no Whittaker vectors_2} the only possible Whittaker vectors in
$ W_{\psi_{1,n}}$ are of the form
$$
\ket{u}=
\underset{\lambda(0)=\lambda(1)=0}{\sum_{|\lambda|=0}}
p_{\lambda} L_{\lambda}\ket{w}.
$$
All vectors of this form satisfy
$$
L_{m} \ket{u}= \psi_{1,n}(L_{n})\delta_{n,m} \ket{u},\;\;\;m\geqslant n
$$
and are therefore Whittaker vectors of the type
$$
\psi_n=\{\psi_{1,n}(L_n),0,\dots,0\}.
$$
We shall call them trivial.

\begin{LL}
\label{lemat_no Whittaker vectors_3}
Let $W_{\psi_{1,n}}$ be the universal Whittaker module $ W_{\psi_{1,n}}$ of a type $\psi_{1,n}$.
The only possible nontrivial Whittaker vectors $\ket{u}\in W_{\psi_{1,n}}$ are of the type $\psi_{1,n}$.
\end{LL}
\begin{proof}
We first show that a vector of the form
\begin{equation}
\label{form}
\ket{u}=
\underset{\lambda(0)=\lambda(1)=0}{\sum_{|\lambda|=0}}
p_{\lambda} L_{\lambda}\ket{w}.
\end{equation}
is not an eigenvector of any $L_k, k=2,\dots,L_{n-1}$.
To this end let us introduce the lexicographic order in the set of partitions of the form
$
\lambda=(\lambda(2),\dots,\lambda(n-1))$:
$$
\lambda<\lambda' \Leftrightarrow
\begin{array}{llllllll}
\lambda(2)<\lambda'(2)\;\vee\;
\\
(\lambda(2)=\lambda'(2)\;\wedge\;\lambda(3)<\lambda'(3))\;\vee\;
\\
\vdots
\\
(\lambda(2)=\lambda'(2)\;\wedge\dots\wedge\; \lambda(n-2)=\lambda'(n-2)\; \wedge\;\lambda(n-1)<\lambda'(n-1)).
\end{array}
$$
For vectors  (\ref{form}) we define
$$
\lambda_{\max}(\ket{u}) = \max \{\lambda \,:\,p_\lambda\neq 0\}.
$$
One easily checks that for $k=2,\dots,n-1$
$$
\lambda_{\max}(L_k \ket{u})>\lambda_{\max}(\ket{u})
$$
hence $\ket{u}$ is not an eigenvector of $L_k$.

 Let
$$
\lambda_{\min}(\ket{u}) = \min \{\lambda \,:\,p_\lambda\neq 0\}.
$$
For each partition $\lambda=(\lambda(2),\dots,\lambda(n-1))$
$$
 \lambda_{\min}(\left[
L_1,  L_{\lambda}
\right]\ket{w})<\lambda
$$
hence
$$
\lambda_{\min}\left(L_1 \ket{u} -\psi_{1,n}(L_1)\ket{u}
\right)
\leqslant
\lambda_{\min}(\left[
L_1,  L_{\lambda_{\min}(\ket{u})}
\right]\ket{w})
<\lambda_{\min}(\ket{u}).
$$
It follows that if $L_1 \ket{u} -\psi_{1,n}(L_1)\ket{u}\neq 0$
it is not proportional to $\ket{u}$.
Thus if
$
L_1\ket{u}= \lambda \ket{u}
$
then $\lambda = \psi_{1,n}(L_1)$.
\end{proof}

\begin{TT}
\label{Whittaker vectors_n34}
Let $\ket{w}$ be a universal Whittaker vector of a type  $\psi_{1,n}$.
\begin{enumerate}
\item
For $n=3$ there are no nontrivial Whittaker vectors in $ W_{\psi_{1,n}}$ of any type.
\item
For $n=4$ the subspace of all nontrivial Whittaker vectors of the type $\psi_{1,n}$ in $ W_{\psi_{1,n}}$
is span by the  family of vectors
\begin{equation}
\label{family4}
\begin{array}{llll}
\ket{w^{\,l}_2}&=& \displaystyle
\sum\limits_{k=0}^l \alpha_k L_2^{l-k}L_3^{2k} \ket{w},\;\;\;l\in \mathbb{N},
\\
\alpha_k &=& - \displaystyle {l+1-k\over 4k\,\psi_{1,4}(L_4)}\,\alpha_{k-1},\;\;\;\alpha_0\neq 0.
\end{array}
\end{equation}
There are no other nontrivial Whittaker vectors  in $ W_{\psi_{1,n}}$ of any type.
\end{enumerate}
\end{TT}
\begin{proof}
By Lemma \ref{lemat_no Whittaker vectors_3} it is enough  to look for nontrivial solutions of the equation
$$
L_{1} \ket{u}= \psi_{1,3}(L_1)) \ket{u}.
$$
For $n=3$ the only possibility is $\ket{u}= \sum p_n L_{2}^n\ket{w}$ for which
$$
(L_1-\psi_{1,3}(L_1))\ket{u}= - \sum n \,p_n  \psi_{1,3}(L_3)L_2^{n-1}\ket{w}.
$$
As the sum on the r.h.s is finite and non vanishing there are no nontrivial Whittaker vectors in $W_{\psi_{1,3}}$ of any type.

For $n=4$ one checks by explicite calculations that vectors (\ref{family4}) are Whittaker vectors of the type $\psi_{1,4}$.
We shall show that for any  Whittaker vector of the type $\psi_{1,4}$ the decomposition
$$
\ket{u}=
{\sum_{\lambda=(\lambda(2),\lambda(3))}}
p_{\lambda} L_{\lambda}\ket{w}.
$$
contains a nonzero term with $\lambda=(k,0)$. By assumption
\begin{equation}
\label{vv}
(L_1-\psi_{1,4}(L_1))\ket{u}={\sum_{\lambda=(\lambda(2),\lambda(3))}}
p_{\lambda} [L_1,L_{\lambda}]\ket{w}=0.
\end{equation}
For each term one has
$$
[L_1,L_2^{\lambda(2)}L_3^{\lambda(3)}]\ket{w}=
\left(-\lambda(2)L_2^{\lambda(2)-1}L_3^{\lambda(3)+1}-2\lambda(3)\psi_{1,n}(L_4)L_2^{\lambda(2)}L_3^{\lambda(3)-1}
\right)\ket{w}.
$$
In order to achieve cancelation of terms on the r.h.s of (\ref{vv}) if the sum contains nonzero term  with
$\lambda=(k,l)$ it has to contain non vanishing terms with $\lambda'=(k-1,l+2)$ and $\lambda''=(k+1,l-2)$.
It follows that the following terms must have nonzero coefficients
$$
L_2^{k'}\ket{w},\;\;\; L_2^{k''}L_3\ket{w}.
$$
The second case however cannot be realized as a non vanishing term in the decomposition of a Whittaker vector of the type $\psi_{1,4}$ since
$$
[L_1,L_2^{k''}L_3]\ket{w}=- k''L_2^{k''-1}L_3^2\ket{w} -2\psi_{1,4}(L_4)L_2^{k''}\ket{w}
$$
and the second term cannot be canceled on the r.h.s. of (\ref{vv}). (This implies in particular that in decomposition (\ref{vv})  $\lambda(3)$
assumes only even values).
It follows that $\ket{u}$ contains at least one term with $\lambda=(k,0)$. Let $\lambda_{\min}=(k_{\min},0)$ be the partition  of this type with smallest $k$.
Then the vector
$$
\ket{u'}=\ket{u}-p_{\lambda_{\min}}\ket{w^{k}_2}
$$
is a new Whittaker vector of the type $\psi_{1,n}$ with $k'_{\min}>k_{\min}$. Hence repeating the subtraction above  a finite number of times one has to get a zero vector.
Thus $\ket{u}$ is a linear combination of vectors (\ref{family4}).

\end{proof}

Construction (\ref{family4}) of Whittaker vectors
can be generalized for $n>4$ as follows
$$
\begin{array}{rlll}
\ket{w^l_{2;n}}&=& \displaystyle\sum\limits_{k=0}^l \alpha_k L_{n-2}^{l-k}L_{n-1}^{2k} \ket{w},\;\;\;l\in \mathbb{N},
\\[15pt]
\alpha_0&\neq& 0,
\\
\alpha_{k+1} &=& - \displaystyle {(n-3)(l-k)\over 2(n-2) (k+1)\,\psi_{1,n}(L_n)}\,\alpha_{k}.
\end{array}
$$
Another generalization  is given by
$$
\begin{array}{rlll}
\ket{w^1_{l;n}}&=&  \displaystyle \sum\limits_{k=\,l}^{n-1} \alpha_k L_{k} L_{n-1}^{k-l} \ket{w},\;\;\;2\leqslant l\leqslant n-2,
\\[15pt]
\alpha_l &\neq&  0,
\\
\alpha_{k+1}   &=& -\displaystyle {k-1 \over (n-2) (k+1-l)\,\psi_{1,n}(L_n)}\,\alpha_{k},\;\;\;
l\leqslant k\leqslant n-3,
\\
\alpha_{n-1}&=&
\displaystyle- {n-3\over (n-2)(n-l)\,\psi_{1,n}(L_n)}\,\alpha_{n-2}.
\end{array}
$$
Constructions above do not exhaust all possibilities. As an illustration we give two more examples for $n=5$
\begin{eqnarray*}
\ket{w^{1,1}_{2,3;5}}&=& \left(
L_2L_3 
-{1\over 3\mu}L_2L_4^2 
-{1\over 3\mu} L_3^2L_4 
+{5\over 3^3 \mu^2} L_3L_4^3 
-{2\over 3^4\mu^3}L^5_4
\right)\ket{w}
\\
\ket{w^2_{2;5}}&=& \left(
L_2^2  -{2\over 3\mu}L_2L_3L_4 +{2^2\over 3^3 \mu^2} L_2L_4^3+{1\over 3^2\mu^2 }L_3^2L_4^2
 -{4\over 3^4\mu^3}L_3L_4^4 -{1\over 3} L_4+{4\over 3^6\mu^4}L_4^6
\right)\ket{w}
\end{eqnarray*}
where $\mu = \psi_{1,5}(L_5)$.
The dimension of the subspace of Whittaker vectors of the type $\psi_{1,n}$ grows
very fast with $n$. A general discussion  is rather
involved and goes beyond the scope of this paper.

We close this section with the theorem characterizing ${\cal V}_{1,n}$ submodules. Its derivation parallels that of
Theorem \ref{tw_podmodul} of the previous section.

\begin{TT}
 Any submodule of a Whittaker module of a type $\psi_{1,n}$ contains a Whittaker vector of the same type.
\label{tw_podmodul_1}
\end{TT}

\section{Gaiotto and BMT states}

We denote by $V_{c,\Delta}$  the Virasoro algebra Verma module of the central charge $c$ and the conformal weight $\Delta$
generated by the highest weight vector
$\ket{\Delta}$:
  \begin{eqnarray*}
z \ket{\Delta} &=& c \ket{\Delta},\;\;\; L_{0}\ket{\Delta} \;=\; \Delta \ket{\Delta},
\;\;\;
 L_{n}\ket{\Delta} \;=\; 0 \hspace{10pt}
   {\rm for} \hspace{10pt} n \geqslant 1.
  \end{eqnarray*}
  Let $V_{c,\Delta}= \bigoplus\limits_{n=0}^\infty V^n_{c,\Delta}$ be the $L_0$ eigenspace decomposition.
In each $n$-level subspace $V^n_{c,\Delta}$ there are two different standard bases :
\begin{eqnarray*}
  &&\{L_{-i_1} \hdots L_{-i_k}\ket{\Delta} : i_1 \geqslant \hdots \geqslant i_k>0, \; \sum_{j=1}^k i_j=n, \; k\in \mathbb{N}  \}\,
  ,
  \\
  &&\{L_{-i_k} \hdots L_{-i_1}\ket{\Delta} : i_1  \geqslant \hdots \geqslant  i_k>0,\; \sum_{j=1}^k i_j=n, \; k\in \mathbb{N}  \}\,
  .
\end{eqnarray*}
We consider the space $V^*_{c,\Delta}$  of all anti-linear forms on $V_{c,\Delta}$.
It is a left $\cal V$ module with respect to the action defined by
 \begin{eqnarray*}
 (L_n f)\ket{u} &=& f(L_{-n}\ket{u})\;
,\;\;\;
 (zf)\ket{u} \,=\, f(z \ket{u}).
 \end{eqnarray*}
Let us denote by $V^{n\,*}_{c,\Delta}$
the space of all  anti-linear forms on  $V^n_{c,\Delta}$.
There are two  bases in $V^{n\,*}_{c,\Delta}$, dual to the standard bases in $V^n_{c,\Delta}$
\begin{eqnarray}
\label{Ibasis}
f^{[k^{n_k},\dots,1^{n_1}]}(L^{m_k}_{-k}\dots L^{m_1}_{-1}\ket{\Delta}) &=& \delta^{n_k}_{m_k}\dots \delta^{n_1}_{m_1},
\\
\label{IIbasis}
f^{[1^{n_1},\dots,k^{n_k}]}(L^{m_1}_{-1} \dots L^{m_k}_{-k}\ket{\Delta}) &=& \delta^{n_k}_{m_k}\dots \delta^{n_1}_{m_1}.
\end{eqnarray}
Each form $f \in V^*_{c,\Delta}$ is uniquely determined by an infinite sequence $\{f_n \}_{n=0}^\infty$
of forms  $f_n\in V^{n\,*}_{c,\Delta}$ which can  in order be  decomposed in one of the dual bases.
This yields two different representations of a form $f \in V^*_{c,\Delta}$ as  the infinite series
\begin{eqnarray}
\label{Arep}
f &=& \sum_{n=0}^\infty \sum_{\sum in_i = n} C_{n_k,\dots,n_1}f^{[k^{n_k},\dots,1^{n_1}]},
\\
\label{Brep}
f &=& \sum_{n=0}^\infty \sum_{\sum in_i = n} D_{n_1,\dots,n_k}f^{[1^{n_1},\dots,k^{n_k}]}.
\end{eqnarray}
Our first aim is to  find    in $V^*_{c,\Delta}$
 all Whittaker vectors  of the pair ${\cal V}_r \subset {\cal V}$
and the type
$$
\psi_r=(\mu_r,\dots,\mu_s,\dots),\;\;r\leqslant s\leqslant 2r,
$$
where $s={\rm rank} \, \psi_r$.
These are the forms $f_{\psi_r}\in V^*_{c,\Delta}$ satisfying
\begin{eqnarray}
 \label{def_whittaker_vector 1}
 L_k f_{ \psi_r} &=& 0\;\;\;\;\;\;\;\;\; {\rm for}\;\; s\leqslant k,
 \\
  \label{def_whittaker_vector 2}
 L_k f_{ \psi_r} &=& \mu_k  f_{ \psi_r}\;\; {\rm for}\;\; r\leqslant k \leqslant s.
 \end{eqnarray}
 Note that by construction  $z f=c f$ for all $f\in V^*_{c,\Delta}$.
 \begin{TT}
 \label{general_form_of_GW_vector}
  A form $f_{ \psi_r} \in V^*_{c,\Delta}$ is a Whittaker vector of the
  pair ${\cal V}_r\subset {\cal V}$  and the type $\psi_r=(\mu_r,\dots,\mu_s,\dots)$
 if and only if  it is of the form
 \begin{equation}
 \label{Th_form_of_whittaker_vector}
 f_{\psi_r} = \sum_{n_{r-1},\dots,n_1=0}^\infty
  A_{n_{r-1}, \hdots, n_{1}}f^{n_{r-1},\hdots, n_1}_{\psi_r},
 \end{equation}
 where
 $$
f^{n_{r-1},\dots, n_1}_{\psi_r} =
\sum_{n_s, \hdots, n_r=0}^\infty \mu_s^{n_s}\dots\mu_r^{n_r}
f^{[s^{n_s}, \dots , r^{n_r},(r-1)^{n_{r-1}}, \dots , 1^{n_1} ]}
$$
and $A_{n_{r-1}, \hdots, n_{1}}$ are arbitrary coefficients.
 \end{TT}

 \begin{proof}
 Let $f_{ \psi_r} \in V^*_{c,\Delta}$ be a form satisfying conditions (\ref{def_whittaker_vector 1}), (\ref{def_whittaker_vector 2}).
 Equations (\ref{def_whittaker_vector 1}) imply that  $f_{ \psi_r} $ vanishes on all basis vectors
 $L_{-i_1} \hdots L_{-i_n} \ket{\Delta}$ involving at least one generator $L_{-n}$ with $n>s$.
Hence $f_{ \psi_r} \in V^*_{c,\Delta}$ is of the following general form
$$
f_{\psi_r} = \sum_{n=0}^\infty \sum_{\sum in_i = n} C_{n_s,\dots,n_1}f^{[s^{n_s},\dots,1^{n_1}]}
=
\sum_{n_s,\dots, n_1=0 }^\infty C_{n_s,\dots,n_1}f^{[s^{n_s},\dots,1^{n_1}]}.
$$
Calculating the left hand  sides of equations (\ref{def_whittaker_vector 2}) on the basis vectors
$L_{-s}^{m_{s}} \hdots L_{-1}^{m_{1}}\ket{\Delta}$
one gets
 \begin{eqnarray*}
 L_i
 && \hspace{-26pt} f_{ \psi_r}
 (L_{-s}^{m_{s}} \hdots L_{-1}^{m_{1}}\ket{\Delta})
 = f_{ \psi_r}(L_{-i} L_{-s}^{m_{s}} \hdots L_{-1}^{m_{1}}\ket{\Delta})
 \\
 &=&
 \sum_{n_s,\dots, n_1=0 }^\infty C_{n_s,\dots,n_1}f^{[s^{n_s},\dots,1^{n_1}]}
 (L_{-s}^{n_{s}} \hdots L_{-i}^{n_{i}+1} \hdots L_{-1}^{n_{1}}\ket{\Delta})
 \\
 &=&
 C_{n_s, \hdots, n_{i}+1, \hdots, n_{1}}.
 \end{eqnarray*}
 Comparing with the right hand sides
 one gets
 $$
 C_{n_s, \hdots, n_{i+1}, \hdots, n_1} = \mu_i C_{n_s, \hdots, n_i, \hdots,n_1}
 $$
 for all $n_1,\dots,n_s$ and $0< i <r$. The most general solution of the conditions above
 takes the form
 $$
C_{n_s, \hdots, n_r,n_{r-1}, \hdots,n_1}= A_{n_{r-1}, \hdots, n_{1}}\mu_s^{n_s} \dots \mu_r^{n_r}
 $$
 with arbitrary coefficients $A_{n_{r-1}, \hdots, n_{1}}$.
 \end{proof}
In the case of BMT states  it is more convenient to work in basis (\ref{IIbasis}).
By similar calculations one gets
\begin{TT}
 \label{general_form_of_BMT_vector}
  A form $f_{ \psi_{1,n}} \in V^*_{c,\Delta}$ is a Whittaker vector of the
  pair ${\cal V}_{1,n}\subset {\cal V}$  and the type $\psi_{1,n}$
 if and only if  it is of the form
 \begin{equation}
 \label{Th_form_of_BTM_vector}
 f_{\psi_{1,n}} = \sum_{m_2,\dots,m_{n-1}=0}^\infty
  B_{m_{2}, \hdots, m_{n-1}}f^{ m_2,\hdots,m_{n-1}}_{\psi_{1,n}},
 \end{equation}
 where
 $$
f^{ m_2,\hdots,m_{n-1}}_{\psi_{1,n}} =
\sum_{m_1, m_n=0}^\infty \nu_1^{m_1}\nu_n^{m_n}
f^{[1^{m_1},2^{m_2}, \dots , (n-1)^{m_{n-1}},n^{m_n}]}
$$
$\nu_1=\psi_{1,n}(L_1)$, $\nu_n=\psi_{1,n}(L_n)$ and $B_{m_{2}, \hdots, m_{n-1}}$ are arbitrary coefficients.
 \end{TT}

Theorems \ref{general_form_of_GW_vector} and \ref{general_form_of_BMT_vector} provide a general
construction  of all Gaiotto and BMT states in terms of finite or infinite combinations of the basic  forms
$f^{n_{r-1},\dots, n_1}_{\psi_r} $
and $f^{ m_2,\hdots,m_{n-1}}_{\psi_{1,n}}$, respectively.
By the results of the Section 2 in the case of the high rank Gaiotto states
all Whittaker modules of a given type are isomorphic. The states can be in principle further
characterized by the transformation properties with respect to the lower generators.
Our basic states are not convenient  from this point of view.
For the simplest Gaiotto states
one has
\begin{eqnarray*}
L_0 f_{ \psi_r}^{0, \hdots, 0}
&=& \Bigl( \Delta + \sum_{l=r}^{s}l\mu_{l} {\partial\over \partial \mu_{l}} \Bigr)
f_{ \psi_r}^{0, \hdots, 0},
\\
L_i f_{ \psi_r}^{0, \hdots, 0}
&=& \sum_{l=r}^{s-i}(l-i)\mu_{i+l} {\partial\over \partial \mu_{l}}f_{ \psi_r}^{0, \hdots, 0},\;\;\;\;1 \leqslant i \leqslant r-1.
\end{eqnarray*}
The corresponding formulae for generic basic forms
contain  several complicated terms and are not especially illuminating.
For the basic BMT states the transformation rules  are much less transparent even in the simplest case.

It is not clear which parts of the enormous spaces of states found in this section are relevant for the AGT relation and the CFT itself.
We are not aware of any explicit construction of higher order Gaiotto states in this context.
For the pair ${\cal V}_{1,n}\subset {\cal V}$ the only examples are the recently discovered BMT states
\cite{Bonelli:2011kv} which in our notation are of the form
$$
\sum_{m_2,\dots,m_{n-1}=0}^\infty
  \lambda_2^{m_{2}} \hdots \lambda_{n-1}^{m_{n-1}}f^{ m_2,\hdots,m_{n-1}}_{\psi_{1,n}}.
$$

\section*{Acknowledgments}

This work was partially financed by means of NCN granted by decision DEC2011/01/B/ST1/01302.


\begin{thebibliography}{99}
\baselineskip 14pt
\parskip 3pt

\bibitem{Arnal74}
D. Arnal and G. Pinczon, {\sl On algebraically irreducible representations of the Lie al-
gebra ${\rm sl}_2$}, J. Math. Phys. {\bf 15} (1974), 350–359.

\bibitem{Block81}
R. Block, {\sl The irreducible representations of the Lie algebra sl(2) and of the Weyl
algebra}, Adv. Math. {\bf 39} (1981),69–110.

\bibitem{Kostant78}
B. Kostant, {\sl On Whittaker vectors and representation theory}, Invent. Math. 48 (1978),
101–184.

\bibitem{McDowell85}
E. McDowell, {\sl On modules induced from Whittaker modules}, J. Algebra {\bf 96} (1985), pp.
161–177.

\bibitem{McDowell93}
E. McDowell, {\sl A module induced from a Whittaker module,} Proc. Am.Math. Soc. {\bf 118}(1993), pp. 349–354

\bibitem{Milicic97}
 D. Mili$\check{\rm c}$i$\acute{\rm c}$ and W. Soergel, {\sl The composition series of modules induced from Whittaker
modules}, Comment. Math. Helv. {\bf 72} (1997), pp. 503–520.


\bibitem{Christodoupoulou08}
K. Christodoupoulou, {\sl Whittaker modules for Heisenberg algebras and imaginary
Whittaker modules for affine lie algebras}, J. Algebra. {\bf 320} (2008),2871–2890.

\bibitem{Benkart09}
G. Benkart and M. Ondrus, {\sl Whittaker modules for generalized Weyl algebras}, Repre-
sent. Theory. {\bf 13} (2009), 141–164.

\bibitem{Ondrus09}
M. Ondrus and E. Wiesner, {\sl Whittaker modules for the Virasoro algebra}, J. Algebra
Appl., {\bf 8} (2009), no. 3, 363 - 377

\bibitem{Ondrus11}
M. Ondrus and E. Wiesner, {\sl Whittaker categories for the Virasoro algebra}
arXiv:1108.2698 [math.RT]

\bibitem{Liu10}
D. Liu, Y. Wu, and L. Zhu, {\sl Whittaker modules for the twisted Heisenberg-Virasoro
algebra}, J. Math. Physics {\bf 51} (2010), (023524, 12).

\bibitem{Zhang10}
X. Zhang, S. Tan,  and H.Lian {\sl Whittaker modules for the Schrödinger–Witt algebra},
J. Math. Phys. {\bf 51} (2010), (083524, 17)

\bibitem{Wang09a}
B. Wang, {\sl Whittaker Modules for Graded Lie Algebras}, Algebr. Represent. Theory {\bf 14}
(2010), 691-702, arXiv:0902.3801 [math.RT]

\bibitem{Wang09b}
B. Wang and J. Li {\sl Whittaker Modules for the $W$-algebra $W(2,2)$}
arXiv:0902.1592v2 [math.RT]

\bibitem{Wang09}
B. Wang, X. Zhu; {\sl Whittaker modules for a Lie algebra of Block type}
 arXiv:0907.0773


\bibitem{Batra}
P. Batra and V. Mazorchuk,
{\sl Blocks and modules for Whittaker pairs}, J. Pure Appl.
Algebra. {\bf 215} (2011), 1552-1568


\bibitem{Gaiotto:2009ma}
  D.~Gaiotto,
  {\sl Asymptotically free N=2 theories and irregular conformal blocks,}
  arXiv:0908.0307 [hep-th].



\bibitem{Alday:2009aq}
  L.~F.~Alday, D.~Gaiotto and Y.~Tachikawa,
  {\sl Liouville Correlation Functions from Four-dimensional Gauge Theories,}
  arXiv:0906.3219 [hep-th].


\bibitem{Braverman:2004vv}
  A.~Braverman,
  {\sl Instanton counting via affine Lie algebras. 1. Equivariant J functions of (affine) flag manifolds and Whittaker vectors,}
  math/0401409 [math-ag].

\bibitem{Braverman:2004cr}
  A.~Braverman and P.~Etingof,
  {\sl Instanton counting via affine Lie algebras II: From Whittaker vectors to the Seiberg-Witten prepotential,}
  math/0409441 [math-ag].


\bibitem{Marshakov:2009gn}
  A.~Marshakov, A.~Mironov and A.~Morozov,
  {\sl On non-conformal limit of the AGT relations,}
  Phys.\ Lett.\  B {\bf 682} (2009) 125
  [arXiv:0909.2052 [hep-th]].




\bibitem{arXiv:1003.1049}
  S.~Yanagida,
  {\sl Whittaker vectors of the Virasoro algebra in terms of Jack symmetric
  polynomial,}
  J. Algebra, 333 (2011), 273-294
  arXiv:1003.1049 [math.QA].
%
\bibitem{Poghossian:2009mk}
  R.~Poghossian,
  {\sl Recursion relations in CFT and N=2 SYM theory,}
  JHEP\ {\bf 0912} (2009) 038
  [arXiv:0909.3412 [hep-th]].
%
\bibitem{Hadasz:2010xp}
  L.~Hadasz, Z.~Jaskolski and P.~Suchanek,
  {\sl Proving the AGT relation for $N_f = 0,1,2$ antifundamentals,}
  JHEP\ {\bf 1006} (2010) 046
  [arXiv:1004.1841 [hep-th]].

\bibitem{Wyllard:2009hg}
  N.~Wyllard,
  {\sl A(N-1) conformal Toda field theory correlation functions from conformal N = 2 SU(N) quiver gauge theories,}
  JHEP\ {\bf 0911} (2009) 002
  [arXiv:0907.2189 [hep-th]].

\bibitem{Mironov:2009by}
  A.~Mironov and A.~Morozov,
  {\sl On AGT relation in the case of U(3),}
  Nucl.\ Phys.\ B\ {\bf 825} (2010) 1
  [arXiv:0908.2569 [hep-th]].


\bibitem{Taki:2009zd}
  M.~Taki,
  {\sl On AGT Conjecture for Pure Super Yang-Mills and W-algebra,}
  JHEP\ {\bf 1105} (2011) 038
  [arXiv:0912.4789 [hep-th]].

\bibitem{Belavin:2011pp}
  V.~Belavin and B.~Feigin,
  {\sl Super Liouville conformal blocks from N=2 SU(2) quiver gauge theories,}
  JHEP\ {\bf 1107} (2011) 079
  [arXiv:1105.5800 [hep-th]].

\bibitem{Bonelli:2011kv}
  G.~Bonelli, K.~Maruyoshi and A.~Tanzini,
  {\sl Instantons on ALE spaces and Super Liouville Conformal Field Theories,}
  arXiv:1106.2505 [hep-th].

\bibitem{Ito:2011mw}
  Y.~Ito,
  {\sl Ramond sector of super Liouville theory from instantons on an ALE space,}
  arXiv:1110.2176 [hep-th].


\bibitem{Nishioka:2011jk}
  T.~Nishioka and Y.~Tachikawa,
  {\sl Central charges of para-Liouville and Toda theories from M-5-branes,}
  Phys.\ Rev.\ D\ {\bf 84} (2011) 046009
  [arXiv:1106.1172 [hep-th]].

\bibitem{Wyllard:2011mn}
  N.~Wyllard,
  {\sl Coset conformal blocks and N=2 gauge theories,}
  arXiv:1109.4264 [hep-th].



%


\bibitem{Keller:2011ek}
  C.~A.~Keller, N.~Mekareeya, J.~Song and Y.~Tachikawa,
  {\sl The ABCDEFG of Instantons and W-algebras,}
  arXiv:1111.5624 [hep-th].


\bibitem{Bonelli:2011aa}
  G.~Bonelli, K.~Maruyoshi and A.~Tanzini,
  {\sl Wild Quiver Gauge Theories,}
  arXiv:1112.1691 [hep-th].



\bibitem{Guo1104}
R. Lu, X. Guo  and  K.  Zhao,
	{\sl Irreducible Modules over the Virasoro Algebra}
	Documenta Math. {\bf 16} (2011) 709--721





\bibitem{Guo1103}
X. Guo and X. Liu; {\sl
Whittaker Modules over Generalized Virasoro Algebras},
Communications in Algebra
{\bf 39} (2011) 3222-3231

\bibitem{Guo1105}
X. Guo and X. Liu; {\sl Whittaker modules over Virasoro-like algebra},
J. Math. Phys. {\bf 52} (2011), (093504, 9 )


\bibitem{MazorchukZhao}
V. Mazorchuk and K. Zhao;
{\sl Simple Virasoro modules which are locally finite over positive part}
arXiv:1205.5937v2 [math.RT]

\end{thebibliography}
\end{document}